\documentclass[11pt]{article}  

\usepackage{amssymb}
\usepackage{amsmath}
\usepackage{amsthm}
\usepackage{color}
\usepackage{colortbl}
\usepackage{graphicx}
\usepackage{hyperref}
\usepackage{fullpage}
\usepackage{xspace}

\usepackage{lmodern}
\usepackage[T1]{fontenc}

\newcommand{\R}{\mathbb{R}}

\DeclareMathOperator*{\E}{\mathbb{E}}

\newtheorem{theorem}{Theorem}
\newtheorem{lemma}[theorem]{Lemma}
\newtheorem{claim}[theorem]{Claim}
\newtheorem{definition}[theorem]{Definition}

\newcommand{\eps}{\epsilon}

\usepackage{caption}
\usepackage{algorithm}
\usepackage{algpseudocode}
\algnewcommand{\LineComment}[1]{\State \(\triangleright\) \textit{#1}}
\algnewcommand{\Annotation}[1]{\State \textcolor{blue}{#1}}

\newcommand{\OPT}{\mathrm{OPT}}

\allowdisplaybreaks

\begin{document}

\title{Submodular Maximization with Nearly-optimal\\ Approximation and Adaptivity in Nearly-linear Time}
\author{
Alina Ene\thanks{Department of Computer Science, Boston University, {\tt aene@bu.edu}. Supported in part by NSF CAREER award 1750333 and NSF CCF award 1718342.}
\and
Huy L. Nguy\~{\^{e}}n\thanks{College of Computer and Information Science, Northeastern University, {\tt hlnguyen@cs.princeton.edu}. Supported in part by NSF CAREER award 1750716.} 
}

\date{}
\maketitle

\begin{abstract}
In this paper, we study the tradeoff between the approximation guarantee and adaptivity for the problem of maximizing a monotone submodular function subject to a cardinality constraint. The adaptivity of an algorithm is the number of sequential rounds of queries it makes to the evaluation oracle of the function, where in every round the algorithm is allowed to make polynomially-many parallel queries. Adaptivity is an important consideration in settings where the objective function is estimated using samples and in applications where adaptivity is the main running time bottleneck. Previous algorithms achieving a nearly-optimal $1 - 1/e - \eps$ approximation require $\Omega(n)$ rounds of adaptivity. In this work, we give the first algorithm that achieves a $1 - 1/e - \eps$ approximation using $O(\ln{n} / \eps^2)$ rounds of adaptivity. The number of function evaluations and additional running time of the algorithm are $O(n \; \mathrm{poly}(\log{n}, 1/\eps))$.
\end{abstract}

\section{Introduction}

The general problem of maximizing a monotone submodular function subject to a size constraint captures many problems of interest both in theory and in practice, including sensor placement, clustering, and influence maximization in social networks. This problem has received considerable attention over the past few decades.  The classical work of Nemhauser, Wolsey, and Fischer \cite{Nemhauser1978} showed that a very natural Greedy algorithm achieves a $1 - 1/e$ approximation for the problem, and this approximation is known to be optimal~\cite{Feige1998,Vondrak2009,Dobzinski2012}. The ensuing decades have led to the development of powerful algorithmic frameworks as well as new applications in areas such as machine learning and data mining.

The Greedy algorithm and its variants play a central role in these developments: Greedy algorithms are natural and simple to use and they achieve the best known approximation guarantees in many settings of interest. The main drawback of Greedy algorithms is that they are inherently sequential and their decisions are intrinsically adaptive.

A recent line of work has focused on addressing the first drawback of Greedy algorithms, and it has led to the development of distributed algorithms for submodular maximization problems in parallel models of computation such as MapReduce~\cite{Kumar2013,Mirzasoleiman2013,MirrokniZ15,Barbosa2015,Mirzasoleiman2015distributed,Barbosa2016,epasto2017bicriteria}. The main focus of these works is on parallelizing sequential algorithms such as Greedy and its variants in order to achieve tradeoffs between the approximation guarantee and resources such as the number of rounds of MapReduce computation and the total amount of communication. In particular, Barbosa \emph{et al.} \cite{Barbosa2016} show that it is possible to achieve a nearly-optimal $1 - 1/e - \eps$ approximation using $O(1/\eps)$ MapReduce rounds. The algorithms developed in these works run a Greedy algorithm on each of the machines, and thus they are just as adaptive as the sequential algorithms. 

Very recently, Balkanski and Singer \cite{BS18} initiated the study of the following question: can we design algorithms for submodular maximization that are less adaptive? The \emph{adaptivity} of an algorithm is the number of sequential rounds of queries it makes to the evaluation oracle of the function, where in every round the algorithm is allowed to make polynomially-many parallel queries:

\begin{definition}[\cite{BS18}]
Given an oracle $f$, an algorithm is {\bf $r$-adaptive} if every query $q$ to the oracle $f$ occurs at a round $i \in [r]$ such that $q$ is independent of the answers $f(q')$ to all other queries $q'$ at round $i$.
\end{definition}

Adaptivity is an important consideration in settings where the objective function is estimated using samples and in applications where adaptivity is the main running time bottleneck. In applications of submodular maximization such as influence maximization and experimental design, queries are experiments that take time and can benefit greatly from parallel execution. In the broader context of optimization and computation, a lot of effort has been devoted to studying the tradeoff of adaptivity and other resources. For example, in property testing, adaptivity has been shown to be crucial with huge gaps in query complexity between non-adaptive and adaptive algorithms and more generally algorithms with different number of adaptive rounds \cite{RS06,CG17}. In compressed sensing, adaptive algorithms can have exponentially fewer measurements than non-adaptive ones (see e.g. \cite{IPW11}). We refer the reader to \cite{BS18} for a more detailed discussion of applications of submodular maximization and the importance of adaptivity in their contexts as well as the study of adaptivity in various areas.

Balkanski and Singer give an algorithm that achieves a $1/3 - \eps$ approximation using $O(\log{n}/\eps^2)$ rounds of adaptivity, and they show that $\Omega(\log{n} / \log\log{n})$ rounds of adaptivity are needed in order to obtain a $1 / \log{n}$ approximation. 

Thus there are now two incomparable algorithms for submodular maximization: the classical Greedy with optimal $1-1/e$ approximation but $O(k)$ adaptivity and the algorithm of \cite{BS18} with $O(\log n/\eps^2)$ adaptivity but $1/3-\eps$ approximation. One cannot help but ask 

\begin{center}
\emph{Is there an inherent tradeoff between adaptivity and approximation?} 
\end{center}

In this work, we obtain an algorithm that is the best of both worlds with nearly optimal approximation $1-1/e-\eps$ and $O(\log{n}/\eps^2)$ adaptivity and $O(n \; \mathrm{poly}(\log{n}, 1/\eps))$ total number of queries and additional running time.

\begin{theorem}
For the problem of maximizing a monotone submodular function subject to a cardinality constraint and any $\eps > 0$, there exists an $O(\log{n} / \eps^2)$-adaptive randomized algorithm which obtains a $1 - 1/e - \eps$ approximation with high probability. The number of function evaluations and additional running time of the algorithm are $O(n \; \mathrm{poly}(\log{n}, 1/\eps))$.
\end{theorem}

{\bf Comparison to \cite{BS18}.}
Let us briefly highlight some of the differences between our work and that of \cite{BS18} (see also Section~\ref{sec:techniques}). The algorithm of \cite{BS18} is based on the single threshold Greedy algorithm, whereas our algorithm is based on the standard Greedy algorithm that achieves the optimal approximation guarantee in the sequential setting. The number of rounds of adaptivity that we obtain matches that of \cite{BS18}, and it is optimal up to lower order terms. Another important point of departure in our algorithm and its analysis is in the running time and function evaluations in each round of adaptivity: we sample only a \emph{poly-logarithmic} number of random sets, and we evaluate the marginal gains only on these random sets, whereas the algorithm of \cite{BS18} uses a \emph{polynomial} number of random sets and it evaluates the marginal gains with respect to all of these random sets. This allows us to obtain an overall nearly-linear running time and function evaluations, which matches up to logarithmic factors the best running time that we can achieve in the sequential setting. The low number of queries as well as their structure make it possible to obtain improved running times for applications such as the ones discussed above.

\subsection{Our techniques}
\label{sec:techniques}

The starting point of our algorithm is the standard Greedy algorithm that achieves the optimal $1 - 1/e$ approximation in the sequential setting. The Greedy algorithm construct the solution sequentially over $k$ iterations, where each iteration adds the element with maximum marginal gain on top of the current solution. An important observation about the standard analysis of Greedy is that the only property that we use about the element selected in each iteration is that its gain $f(S \cup \{e\}) - f(S)$ is at least $\frac{1}{k} (f(\OPT) - f(S))$. Following \cite{BS18}, to achieve a low adaptivity, we want to add not just a single element but a much larger set of elements in each round of adaptivity. In contrast to \cite{BS18}, which based their approach on the single threshold Greedy algorithm, we draw inspiration from the standard Greedy algorithm and its analysis. Thus we aim to add a large set $R$ in each iteration whose density (ratio of gain to size) nearly matches the density of the optimal solution, i.e., we have $\frac{f(R \cup S) - f(S)}{|R|} \geq (1 - O(\eps)) \frac{f(\OPT) - f(S)}{k}$.

Perhaps surprisingly, we show that we can implement this strategy using $O(\ln{n})$ rounds of adaptivity, which matches up to lower order terms the hardness result of \cite{BS18}. The algorithm leverages the following dichotomy inspired by the work of \cite{BS18} and the earlier work on sample and prune of \cite{Kumar2013}. Suppose that we choose the set $R$ by sampling suitably many elements uniformly at random. If the expected density of the random set is almost as high as the target density of $\frac{f(\OPT) - f(S)}{k}$, then we can add the random set to our solution and gain as much as Greedy. On the other hand, if the density is low, we are guaranteed that a constant fraction of the elements have low expected marginal gain. Using this insight, we make progress by filtering the elements with low marginal gain. Since each filtering step removes a constant fraction of the elements, we ensure that we have only $O(\ln{n})$ rounds of adaptivity. At the same time, we are able to argue that the filtering steps preserve most of the value of $\OPT$, which is essential for obtaining the nearly-optimal approximation guarantee.

We also ensure that the overall running time and function evaluations of our algorithm is nearly-linear. This requires new insights and a different analysis from that of \cite{BS18}. In each adaptive round, the algorithm needs to estimate the expected gain of a random set and the expected marginal gain of each element by sampling sufficiently many random sets. The analysis of \cite{BS18} relies on having accurate estimates for the marginal gain of each element. Since some of the marginal gains can be very small (even among the elements of the optimal solution, most marginal gains might be as small as $\frac{1}{k}\cdot f(\OPT)$), it is necessary to sample a \emph{polynomial} number of random sets to ensure that every gain is estimated accurately enough. We take a very different approach that allows us to sample only a \emph{poly-logarithmic} number of random sets.

A key difficulty is to ensure that the filtering does not remove too much value from $\OPT$. The crucial insight here is the following. Evaluating the marginal gain over a common set (random or otherwise) can significantly decrease the marginal gains of elements, but the \emph{aggregate value} decreases by at most the value of the common set itself. Since we only filter when the random set has small value, even though some elements in $\OPT$ appear to have very small marginal gain with respect to the random set and they get filtered, the overall decrease in value can be charged to the random set and therefore it is low.

\medskip\noindent
{\bf Independent work.} Finally, we note two independent results with the same approximation and number of rounds of adaptivity~\cite{BRS18,FMZ18}. The algorithm of \cite{BRS18} is similar to ours. The main difference between the two algorithms is in the number of random sets used in each round of adaptivity. Our algorithm uses a poly-logarithmic number of random sets and has an overall nearly-linear running time. In contrast, the algorithm of \cite{BRS18} follows the approach of \cite{BS18} and uses a polynomial number of random sets and an analysis based on estimating the marginal values. As discussed above, we use a different analysis to handle the much smaller number of random sets. The algorithm of \cite{FMZ18} makes a linear number of queries in expectation.

\subsection{Preliminaries and notation}

Let $f: 2^V \rightarrow \R_{\geq 0}$ be a set function on a ground set $V$ of size $n = |V|$. The function is \emph{submodular} if $f(A) + f(B) \geq f(A \cap B) + f(A \cup B)$ for all subsets $A, B$. The function is \emph{monotone} if $f(A) \leq f(B)$ for all subsets $A, B$ satisfying $A \subseteq B$.

We consider the problem of maximizing a monotone submodular function subject to a cardinality constraint: find $S^* \in \arg\max_{S \subseteq V \colon |S| \leq k} f(S)$.

For any two sets $A$ and $B$, we use the notation $f(A | B)$ to denote the marginal gain of $A$ on top of $B$, i.e., $f(A | B) = f(A \cup B) - f(B)$. For an element $e \in V$, we use $f(e | B)$ as a shorthand for $f(\{e\} | B)$.

We will use the following Chernoff inequality, which follows from the standard Chernoff bound (see, e.g., \cite{dubhashi2009concentration}). 

\begin{theorem}[\cite{dubhashi2009concentration}]
\label{thm:chernoff}
Let $X_1, \dots, X_n$ be mutually independent and identically distributed random variables with $X_i \in [0, 1]$. Let $X = \frac{1}{n}\sum_{i = 1}^n X_i$. Suppose that $\E[X] \leq \mu_H$. Then, for every $0 < \eps < 1$, we have
  \[ \Pr[X > (1 + \eps) \mu_H] \leq \exp\left( - \frac{\eps^2}{3} n \mu_H \right).\]
\end{theorem}

\section{Submodular Maximization Algorithm}

\begin{algorithm}[h]
\caption{The input is a submodular function $f: 2^V \rightarrow \R_{\geq 0}$ that is monotone and non-negative, a cardinality constraint $k$, and an error parameter $\eps$. }
\label{algorithm2}
\begin{algorithmic}[1]
\State For a set $U \subseteq V$ and an integer $\ell \leq |U|$, we let $\mathcal{U}(U, \ell)$ be the uniform distribution over subsets of $U$ of cardinality $\ell$
\State $M$ is an approximate optimal solution value: $M \le f(\mathrm{OPT}) \le (1+\eps)M$
\State $m=O( (\ln{n})^2 / \eps^4)$, $\ell = \eps^2 k/(100\ln n)$
\State $S = \emptyset$
\While{$f(S) \le (1-1/e-O(\eps))M$} \label{line:begin-phase}
  \State $U_1 = V\setminus S$ \Comment{Unfiltered elements}
  \State $t = 0$
  \State $Old = f(S)$
  \While{$f(S) - Old < \eps M/100$} \label{line:begin-iteration}
    \State $t \gets t+1$
    \State Let $R_1, \ldots, R_m$ be independent samples from $\mathcal{U}(U_t, \ell)$
    \State Let $R_{max} = \arg\max_{R\in \{R_1,\ldots, R_m\}} f(R|S)$
  \If{$f(R_{max}|S) \ge (1-10\eps)\frac{\ell}{k}(M-Old)$} \label{line:addR-branch}
    \State $S \gets S\cup R_{max}$ \label{line:add-R}
    \State $U_{t+1} = U_t\setminus R_{max}$
  \Else \label{line:filter-branch}
    \State Let $\{R_{i, j} \colon 1 \leq i \leq \Theta(\ln{n} / \eps), 1 \leq j \leq m\}$ be independent samples from $\mathcal{U}(U_t, \ell)$
    \State Let $v_{i, e} = \frac{1}{m} \sum_{j = 1}^m f(e|S\cup R_{i,j})$
    \State Let $avg_{i} = \frac{1}{m} \sum_{j = 1}^m f(R_{i,j} | S)$
    \State Let $i$ be s.t. $avg_i \leq (1 - 8\eps)\frac{\ell}{k} (M - Old)$ and $\sum_{e \in U_t} v_{i, e} \leq |U_t| (1 - 8\eps)\frac{M - Old}{k}$ \label{line:good-batch}
    \State If there is no such $i$, declare failure and terminate \label{line:fail}
    \State $U^-_t = \{e\in U_t: v_{i,e} < (1-7\eps) \frac{M-Old}{k}\}$ \label{line:filter}
    \State $U_{t+1} = U_t \setminus U^-_t$
  \EndIf
  \EndWhile
\EndWhile
\State \Return $S$
\end{algorithmic}
\end{algorithm}

The algorithm is given in Algorithm~\ref{algorithm2}. We assume that the algorithm has access to a value $M$ such that $M \leq f(\mathrm{OPT}) \leq (1+\eps)M$. An $n$-approximation to $f(\mathrm{OPT})$ is $M_0 = \max_{e\in V} f(\{e\})$. Given this value, we can try $2\eps^{-1}\ln n$ guesses for $M$: $M_0, (1+\eps)M_0, (1+\eps)^2 M_0, \ldots$ in parallel and return the best solution from all the guesses.

The algorithm builds the solution over $1/\eps$ phases, where each phase increases the solution value by $\Omega(\eps f(\OPT))$; a phase corresponds to a single iteration of the while loop on line~\ref{line:begin-phase}.  In each iteration of a phase (an iteration of the while loop on line~\ref{line:begin-iteration}), we aim to find a set with density $(1 - O(\eps)) \frac{f(\OPT) - Old}{k}$ and size $\Theta\left(\frac{\eps^2}{\ln{n}}\right) \cdot k$, where $Old$ is the value $f(S)$ of the solution at the beginning of the phase. If the expected density of a random set is at least the target density then, by sampling enough random sets, we can guarantee that with high probability we find a good set to add to our solution (see Lemma~\ref{lemma1}). On the other hand, if the expected density of a random set is below the target, then we can show that an $\eps$ fraction of the elements have expected marginal gain at most $1 + O(\eps)$ times the target density (see Lemma~\ref{lemmaR}). Thus, by sampling enough random sets, we can guarantee that with high probability we filter an $\eps$ fraction of the elements on line~\ref{line:filter}.

Another key issue is determining how many random sets we need to sample. Since we are filtering based on marginal gains, it is tempting to proceed by ensuring that each marginal gain is estimated to sufficient accuracy. Unfortunately, this requires sampling polynomially many random sets. To obtain a fast running time, we take a different approach that uses only a poly-logarithmic number of random sets. Since each random set has many elements, we can show using a Chernoff bound argument that our estimate of the expected value $\E_{R}[f(R|S)]$ of the random set is correct with high probability (see Lemma~\ref{lemma1}). When the expected value of the random set is small, it holds deterministically that the average of the expected marginal gains of the elements is small (see Claim~\ref{claim2}). This fact together with a straightforward application of Markov's inequality and the Chernoff inequality gives us that with high probability the algorithm executes the filtering step on line~\ref{line:filter} (see Lemma~\ref{lemma1}). This ensures that, when the expected value of the random set is low, we filter many elements with high probability. We also need to argue that we do not filter too much of the value of $\OPT$. Here we cannot rely on the marginal values being estimated accurately and we need a different analysis. A crucial insight is that we can analyze the loss on aggregate. A key idea is to track the value $f((\OPT \cap U_t) \cup S)$ of the optimal solution $\OPT \cap U_t$ that has survived the filtering steps so far. Instead of relying on having accurate marginal gains, we use a \emph{deterministic} analysis to bound the loss in the value $f((\OPT \cap U_t) \cup S)$ in a filtering round: since we are evaluating the marginal gains on common sets that have low value, even though some elements of $\OPT$ appear to have very small marginal gain with respect to these sets and they get filtered, the overall decrease in value is at most the value of the random sets themselves (see Claim~\ref{invariant}).  

\subsection{The analysis of Algorithm~\ref{algorithm2}}

We divide the execution of the algorithm into phases corresponding to the iterations of the outer while loop. We will show that the number of phases is bounded by $O(1/\eps)$ and the number of iterations of the inner while loop in each phase is $O(\ln{n}/\eps)$. Therefore, the total number of rounds of adaptivity is $O(\ln{n}/\eps^2)$.

Consider an iteration of the inner while loop (line~\ref{line:begin-iteration}). We show that if $\E_{R\sim \mathcal{U}(U,\ell)}[f(R|S)] \ge (1 - 9 \eps) \frac{\ell}{k} (M-Old)$ then with high probability, the algorithm executes line~\ref{line:add-R} (see Lemma~\ref{lemmaR}). On the other hand, if $\E_{R\sim \mathcal{U}(U,\ell)}[f(R|S)] < (1 - 9\eps) \frac{\ell}{k}(M-Old)$ then, with high probability, the elements filtered on line~\ref{line:filter} account for at least an $\eps$ fraction of all elements in $U$ (see Lemma~\ref{lemma1}).

\begin{lemma}
\label{lemmaR}
Consider an iteration $t$ with $\E_{R \sim \mathcal{U}(U_t, \ell)}[f(R|S)] \ge \frac{\ell}{k}(1-9\eps)(M-Old)$. With probability $1-1/n^2$, we have $f(R_{max}|S) \ge \frac{\ell}{k}(1-9\eps)(M-Old)$ and the algorithm executes line \ref{line:add-R}.
\end{lemma}
\begin{proof}
Using Markov's inequality, we will show that a given random set has a high value with probability at least $\Omega(\eps^3 / \ln{n})$. Since we are independently sampling $m = \Theta(\ln^2{n} /\eps^4)$ sets, at least one of the sets has a high value with high probability.

Note that, for a random set $R \sim \mathcal{U}(U_t, \ell)$, we have $0 \le f(R|S) \le f(\mathrm{OPT}) \le (1+\eps)M$: the first inequality follows by monotonicity, and the second inequality follows from the fact that $f(R | S) \leq f(R) \leq f(\mathrm{OPT})$, since $R$ is feasible. Since $(1 + \eps) M - f(R | S)$ is a non-negative random variable, it follows from Markov's inequality that
\begin{align*}
& \Pr_{R \sim \mathcal{U}(U_t, \ell)}\left[ (1 + \eps) M - f(R | S) > (1 + \eps) M - \frac{\ell}{k} (1 - 10\eps)(M - Old) \right]\\
&\quad \leq \frac{\E_{R \sim \mathcal{U}(U_t, \ell)} \left[ (1 + \eps) M - f(R | S)\right]}{(1 + \eps) M - \frac{\ell}{k} (1 - 10\eps)(M - Old)}\\
&\quad \leq \frac{ (1 + \eps) M - \frac{\ell}{k}(1-9\eps)(M-Old)}{(1 + \eps) M - \frac{\ell}{k} (1 - 10\eps)(M - Old)}\\
&\quad = 1 - \frac{\frac{\ell}{k} \eps (M-Old)}{(1 + \eps) M - \frac{\ell}{k} (1 - 10\eps)(M - Old)}\\
&\quad = 1 - \Theta\left(\frac{\eps^3}{\ln{n}} \right)
\end{align*}
The second inequality is our assumption, and the last equality follows from the fact that $\ell/k = \Theta(\eps^2 / \ln{n})$ and $(1/e + O(\eps)) M \leq M - Old \leq M$.

Therefore, with probability at least $\Omega(\eps^3 / \ln{n})$, we have $f(R|S) \geq \frac{\ell}{k} (1 - 10\eps)(M - Old)$. Since we independently sample $m = \Theta((\ln{n})^2 / \eps^4)$ sets, with probability at least $1-1/n^2$, we find a set $R_{max}$ such that $f(R_{max}|S) \ge \frac{\ell}{k}(1-10\eps)(M-Old)$.
\end{proof}

We now consider the case when the expected value of the random set is below the target and show that, with high probability, the algorithm filters many elements on line~\ref{line:filter}.

\begin{lemma}
\label{lemma1}
Consider an iteration $t$ in which $\E_{R \sim \mathcal{U}(U_t, \ell)}[f(R|S)] \le (1 - 9\eps) \frac{\ell}{k}(M-Old)$ and the algorithm executes line~\ref{line:filter-branch}. With probability at least $1 - 1/n^3$, the algorithm does not fail on line~\ref{line:fail}. Additionally, if the algorithm does not fail then $|U^-_t| \geq \eps |U_t|$.
\end{lemma}
\begin{proof}
We first give an overview of the proof. In Claim~\ref{claim2}, we show that the expected density $\frac{\E_R[f(R|S)]}{|R|}$ of the random set is at least the average expected marginal gain $\frac{1}{|U_t|} \sum_{e \in U_t} \E_R[f(e | S \cup R)]$ of the elements. Note that this is a claim about expected values and thus it holds deterministically.  Using a Chernoff inequality (Theorem~\ref{thm:chernoff}), we show that, with high probability, we correctly determine that the expected value $\E_R[f(R|S)]$ of the random set is below the target. Additionally, we show that, with high probability, the algorithm succeeds to determine that the average marginal gain of the elements is low: by Claim~\ref{claim2} and Markov's inequality, a single random set succeeds with constant probability; since we independently sample poly-logarithmically many random sets, we obtain high probability overall. Thus, with high probability, the algorithm executes the filtering step on line~\ref{line:filter} and a straightforward averaging argument shows that it filters an $\eps$ fraction of the elements. 

\begin{claim}
\label{claim2}
We have
\[\E_{R\sim \mathcal{U}(U_t,\ell)}[f(R|S)] \geq \ell \cdot \frac{1}{|U_t|} \sum_{e \in U_t} \E[f(e|S\cup R)].\]
\end{claim}
\begin{proof}
Consider a random set $R$. Order the elements of $R$ arbitrarily as $e_1, e_2, \dots, e_{\ell}$ and let $R_0 = \emptyset$ and $R_i = \{e_1, \dots, e_i\}$. 
\[ f(R | S) = f(S \cup R) - f(S)
= \sum_{i = 1}^{\ell} (f(S \cup R_i) - f(S \cup R_{i - 1}))
= \sum_{i = 1}^{\ell} f(e_i | S \cup R_{i - 1})
\geq \sum_{i = 1}^{\ell} f(e_i | S \cup (R \setminus \{e_i\})) \]
where the last inequality follows from submodularity.

Therefore we have
\begin{align}
\E_{R\sim \mathcal{U}(U_t, \ell)}[f(R|S)] &\ge \E\left[\sum_{e\in R} f(e|S\cup (R\setminus \{e\})) \right] \notag\\
&=\sum_e \Pr[e\in R] \cdot \E[f(e|S\cup (R\setminus \{e\})) | e\in R] \notag\\
&=\sum_e \frac{\ell}{|U_t|} \cdot \E[f(e|S\cup (R\setminus \{e\})) | e\in R] \label{eq1}
\end{align}
Let us now show that, for every $e$, we have
\[ \E_{R \sim \mathcal{U}(U_t, \ell)}[f(e | S \cup (R \setminus \{e\})) | e \in R] \geq \E_{R \sim \mathcal{U}(U_t, \ell)}[f(e | S \cup (R \setminus \{e\})) | e \notin R]\]
To this end, note that we may assume that $R \sim \mathcal{U}(U_t, \ell)$ is generated by choosing a permutation $\pi$ of $U_t$ uniformly at random and letting $R = \{e_{\pi_1}, e_{\pi_2}, \dots, e_{\pi_{\ell}}\}$ be the first $\ell$ elements in this permutation. We have
\begin{align*}
\E_{R \sim \mathcal{U}(U_t, \ell)}[f(e|S\cup (R\setminus \{e\})) | e\in R]
&= \E_{R' \sim \mathcal{U}(U_t \setminus \{e\}, \ell - 1)}[f(e|S \cup R')] \\
&= \E_{R \sim \mathcal{U}(U_t, \ell)}[f(e | S \cup (R \setminus \{e, e_{\pi_1}\}) | e \notin R] \\
&\geq \E_{R \sim \mathcal{U}(U_t, \ell)}[f(e | S \cup (R \setminus \{e\}) | e \notin R]
\end{align*}
In the first equality, we have used that, if $R \sim \mathcal{U}(U_t, \ell)$ and $e \in R$, then $R \setminus \{e\}$ has the distribution $\mathcal{U}(U_t \setminus \{e\}, \ell - 1)$. In the second equality, we have used that, if $R \sim \mathcal{U}(U_t, \ell)$ and $e \notin R$, then $R \setminus \{e, e_{\pi_1}\}$ has the distribution $\mathcal{U}(U_t \setminus \{e\}, \ell - 1)$, since $e_{\pi_1}$ is an element of $R$. The inequality follows by submodularity.

Therefore
\begin{align}
& \E_{R \sim \mathcal{U}(U_t, \ell)}[f(e | S \cup (R \setminus \{e\}))] \notag\\
& \quad = \E[f(e | S \cup (R \setminus \{e\})) | e \in R] \Pr[e \in R] + \E[f(e | S \cup (R \setminus \{e\})) | e \notin R] \Pr[e \notin R] \notag\\
&\quad \leq \E[f(e | S \cup (R \setminus \{e\})) | e \in R] (\Pr[e \in R] + \Pr[e \notin R]) \notag\\
&\quad = \E[f(e | S \cup (R \setminus \{e\}) | e \in R] \label{eq2}
\end{align}
By combining (\ref{eq1}) and (\ref{eq2}), and using submodularity, we obtain
\[
\E_{R\sim \mathcal{U}(U_t,\ell)}[f(R|S)] \ge \sum_e \frac{\ell}{|U_t|} \cdot \E[f(e|S\cup (R\setminus \{e\}))] \ge \sum_e \frac{\ell}{|U_t|} \cdot \E[f(e|S\cup R)]
\]
\end{proof}

Let us now show that, with probability $1-1/n^3$, there is a batch of random sets $\{R_{i, j} \colon j \in [m]\}$ with the properties stated on line~\ref{line:good-batch}.

Fix a batch $i$. Using Theorem~\ref{thm:chernoff}, we can upper bound the probability of the event that $avg_i > (1 - 8\eps) \frac{\ell}{k} (M - Old)$ as follows. For each $j \in [m]$, let $X_j = \frac{1}{(1 + \eps)M} f(R_{i, j} | S) \in [0, 1]$. Let $X = \frac{1}{m} \sum_{j = 1}^m X_j$. By our assumption, we have
\[ \E_{R \sim \mathcal{U}(U_t, \ell)}[X_j] \leq \frac{1}{(1 + \eps) M }(1 - 9\eps) \frac{\ell}{k} (M - Old).\]
Let $\mu_H = \frac{1}{(1 + \eps) M} (1 - 9\eps) \frac{\ell}{k} (M - Old)$. By Theorem~\ref{thm:chernoff},
\[ \Pr[X > (1 + \eps) \mu_H] \leq \exp\left( - \frac{\eps^2}{3} m \mu_H \right) \leq \frac{1}{n^3},\]
where the second inequality follows by substituting $m$ and $\mu_H$, and using the fact that $(M - Old) / M = \Theta(1)$.

We now upper bound the probability of the event that $\sum_{e \in U_t} v_{i, e} > |U_t| (1 - 8\eps)\frac{M - Old}{k}$. By Markov's inequality, with probability at least $\eps$, we have
\begin{align*}
\frac{1}{m} \sum_{j = 1}^m \left(\sum_e f(e|S \cup R_{i,j}) \right) &\le \frac{1}{1-\eps} \E_{R\sim \mathcal{U}(U_t, \ell)}\left[ \sum_e f(e|S \cup R) \right]\\
&\le \frac{1}{1-\eps} \frac{|U_t|}{\ell} \E_{R\sim \mathcal{U}(U_t, \ell)}[f(R|S)]\\
&\le \frac{1}{1-\eps} \frac{|U_t|}{\ell} \frac{\ell}{k}(1-9\eps)(M-Old)\\
&\le |U_t|(1-8\eps)\frac{M-Old}{k}
\end{align*}
In the second inequality, we have used Claim~\ref{claim2}.

Therefore each batch $i$ satisfies both conditions of line~\ref{line:good-batch} with probability at least $\eps - 1/n^3$. Since the batches are independent, it follows that the probability that the algorithm fails on line~\ref{line:fail} is at most $(1 - \eps + 1/n^3)^{\Theta(\ln{n}/\eps)} \leq 1/n^3$.

Let us now condition on the event that the algorithm does not fail. Consider the set $U^-_t$ filtered on line~\ref{line:filter}. We have
\[ |U_t \setminus U^-_t| (1 - 7\eps) \frac{M - Old}{k}  \leq \sum_{e \in U_t} v_{i, e} \leq |U_t| (1 - 8\eps)\frac{M - Old}{k},\]
and thus $|U^-_t| \geq \eps |U_t|$.
\end{proof}

We now show that the number of phases is $O(1/\eps)$ and the number of iterations in each phase is $O(\ln{n}/\eps)$. The former simply follows from the fact that each phase increases the value of the solution by $\Omega(\eps f(\OPT))$. Most of the work is to show that the filtering steps do not remove too much of the optimal solution. A subtle but crucial choice is to track the value of the optimal solution $\OPT \cap U_t$ that has survived the filtering steps so far. In Claim~\ref{invariant}, we analyze how much this value decreases in each filtering iteration and show that this decrease can be charged to the value of the random sets, which have low value. Claim~\ref{invariant} then allows us to show that we cannot keep filtering without eventually finding a good set to add on line~\ref{line:add-R}: since each filtering iteration removes an $\eps$ fraction of the elements, after $O(\ln{n}/\eps)$ filtering iterations the ground set becomes empty; on the other hand, Claim~\ref{invariant} shows that $O(\ln{n}/\eps)$ filtering iterations is not enough to remove all of the value of $\OPT$, since $f(\OPT \cap U_t | S)$ is strictly positive (see Claim~\ref{filter-iterations}). 

\begin{lemma}
\label{lemma3}
Consider a phase of the algorithm. The phase increases $f(S)$ by $\Omega(\eps M)$. Additionally, with probability at least $1 - 1/n^2$, the phase has $O(\ln{n}/\eps)$ iterations.
\end{lemma}
\begin{proof}
The lower bound on the increase follows from the terminating condition for the phase. Thus it only remains to bound the number of iterations. We refer to each iteration as a gain iteration if line~\ref{line:addR-branch} is executed, and as a filtering iteration if line~\ref{line:filter-branch} is executed. We show that the number of gain iterations is $O(\ln{n}/\eps)$ with probability $1$, and the number of filter iterations is $O(\ln{n}/\eps)$ with probability at least $1 - 1/n^2$.

\begin{claim}
\label{gain-iterations}
The number of gain iterations is at most $6\ln{n}/\eps$.
\end{claim}
\begin{proof}
Each gain iteration increases $f(S)$ by $\frac{\ell}{k} (1 - 10\eps)(M - Old)$. Since $M - Old \geq M/3$ and the phase ends when $f(S) - Old$ becomes $\eps M / 100$, the number of gain iterations is at most $6\ln{n}/\eps$.
\end{proof}

Let us now consider the filtering iterations. Let $T$ be the minimum of $2\ln{n}/\eps + 1$ and the number of filtering iterations of the phase. By Lemma~\ref{lemma1}, the probability that none of the first $T$ filtering iterations fails is at least $1 - T/n^3 \geq 1 - 1/n^2$. In the following, we condition on the event that none of the first $T$ filtering iterations fails.

We can show the following invariant.

\begin{claim}
\label{invariant}
Consider an iteration $t$ and suppose that the number of filtering iterations so far is at most $T$. At the beginning of iteration $t$, we have
\[ f((\mathrm{OPT} \cap U_t)\cup S) \ge M - \frac{(t-1)\ell}{k} (M-Old) - \frac{|\mathrm{OPT}\setminus (U_t\cup S)|}{k}(1-7\eps)(M-Old). \]
\end{claim}
\begin{proof}
We will prove the invariant by induction on $t$. The invariant is true at the beginning of iteration $1$ since $f((\mathrm{OPT}\cap U_1)\cup S) \ge f(\mathrm{OPT}) \ge M$.

Consider iteration $t > 1$ and suppose the invariant holds at the beginning of iteration $t$. We will show that the invariant continues to hold at the beginning of iteration $t + 1$.

Suppose that iteration $t$ is a gain iteration. The algorithms adds the random set $R_{max}$ to $S$ on line~\ref{line:add-R} and we have
\[ f((\mathrm{OPT} \cap (U_{t}\setminus R_{max})) \cup S \cup R_{max}) \geq f((\mathrm{OPT} \cap U_t) \cup S),\]
since $(\mathrm{OPT} \cap (U_{t}\setminus R_{max})) \cup S\cup R_{max} \supseteq (\mathrm{OPT} \cap U_t) \cup S$ and $f$ is monotone. Thus the invariant continues to hold at the beginning of iteration $t+1$.

Therefore we may assume that iteration $t$ is a filtering iteration. Recall that we are conditioning on the event that the algorithm does not fail in the first $T$ filtering iterations, and thus iteration $t$ executes line~\ref{line:filter}. Let $i$ be the index satisfying the conditions on line~\ref{line:good-batch}. We have
\begin{align*}
& f(\mathrm{OPT} \cap U_{t + 1} | S)\\
&\quad \geq \frac{1}{m} \sum_{j = 1}^m f(\mathrm{OPT} \cap U_{t+1}| S \cup R_{i,j})\\
&\quad \geq \frac{1}{m} \sum_{j = 1}^m \Bigg(f(\mathrm{OPT} \cap U_{t} | S\cup R_{i,j}) - \sum_{e\in \mathrm{OPT} \cap U^-_t} f(e|S\cup R_{i,j}) \Bigg)\\
&\quad \ge f(\mathrm{OPT} \cap U_t | S) - \frac{1}{m} \sum_{j = 1}^m f(R_{i,j}|S) - |\mathrm{OPT}\cap U^-_t| (1-7\eps)\frac{1}{k}(M-Old)\\
&\quad \ge  f(\mathrm{OPT}\cap U_t|S) - \frac{\ell}{k}(M-Old) -|\mathrm{OPT}\cap U^-_t| (1-7\eps)\frac{1}{k}(M-Old)
\end{align*}
The first and second inequalities follow by submodularity. In the third inequality, we have used monotonicity to bound $f((\mathrm{OPT} \cap U_t) \cup S \cup R_{i,j}) \geq f((\mathrm{OPT} \cap U_t) \cup S)$, and we have used that $v_{i,e} \leq  (1-7\eps)\frac{1}{k}(M-Old)$ for all $e \in U^-_t$. In the fourth inequality, we have used that $avg_i \leq \frac{\ell}{k} (1 - 8\eps) (M - Old)$.

It follows that the invariant continues to hold at the beginning of iteration $t+1$.
\end{proof}

\begin{claim}
\label{filter-iterations}
The number of filtering iterations is at most $2\ln{n}/\eps$.
\end{claim}
\begin{proof}
Recall that we are conditioning on the event that the first $T$ filtering iterations do not fail. Suppose for contradiction that the number of filtering iterations reaches $2\ln{n}/\eps + 1$, and let $t$ be the iteration when this happens. By Lemma~\ref{lemma1}, each filtering iteration removes an $\eps$ fraction of $U$ and thus $U_t = \emptyset$. By Claim~\ref{gain-iterations}, the number of gain iterations is at most $6\ln{n}/\eps$ and thus $t \leq 8\ln{n}/\eps + 1$. By Claim~\ref{invariant}, at the beginning of iteration $t$, we have

\begin{align*}
& f((\mathrm{OPT} \cap U_t)\cup S) - f(S)\\
&\quad \geq M - \frac{(t-1)\ell}{k} (M-Old) - \frac{|\mathrm{OPT}\setminus (U_t\cup S)|}{k}(1-7\eps)(M-Old) - f(S)\\
&\quad \geq M - \frac{8\eps}{100} (M - Old) - (1 - 7\eps)(M - Old) - f(S)\\
&\quad = \left(7\eps - \frac{8\eps}{100}\right) (M - Old) - (f(S) - Old)\\
&\quad \geq \left(7\eps - \frac{8\eps}{100}\right) \cdot \frac{1}{3} M - \frac{\eps}{100} M\\
&\quad > 0
\end{align*}
In the second inequality, we used that $|\mathrm{OPT} \setminus (S_t \cup S)| \leq k$. In the third inequality, we used that $M - Old \geq (1/e + O(\eps)) M \geq 1/3 M$. In the last inequality, we used the fact that $f(S) - Old < \eps M /100$, since the phase has not ended.

It follows that $U_t$ is non-empty, which is a contradiction. Therefore the number of filtering phases is at most $2\ln{n}/\eps$.
\end{proof}

\end{proof}

\begin{lemma}
With probability $1-1/n$, the algorithm uses $O(\ln n / \eps^2)$ rounds of queries.
\end{lemma}
\begin{proof}
The lemma follows from Lemma~\ref{lemma3}. Each phase increases $f(S)$ by $\Omega(\eps M)$ and the algorithm stops when $f(S) \ge (1-1/e-O(\eps))M$. Thus, the number of phases is $O(1/\eps)$. With probability $1 - \frac{1}{\eps n^2} \geq 1 - \frac{1}{n}$, every phase uses $O(\ln{n}/\eps)$ rounds of queries. Therefore, with probability $1-1/n$, the total number of rounds of queries is $O(\ln{n} / \eps^2)$.
\end{proof}

Finally, we argue that the algorithm returns a feasible solution and achieves a $1-1/e-O(\eps)$ approximation ratio.

\begin{lemma}
The algorithm returns a feasible solution and achieves a $1 - 1/e - O(\eps)$ approximation.
\end{lemma}
\begin{proof}
Consider iteration $j$ of a phase where the algorithm executes line~\ref{line:add-R}.
We have
\begin{align*}
f(S\cup R_{max})-f(S) &\ge |R_{max}|(1-10\eps)(M-Old)/k\\
&\ge |R_{max}|(1-11\eps)(M-f(S))/k
\end{align*}
Using the inequality above, we can show by induction that
\[ M - f(S) \leq \exp\left(- \frac{(1-11\eps)|S|}{k}\right) M.\]
Initially, $S = \emptyset$ and the inequality holds. Consider an iteration where the inequality holds at the beginning of the iteration. We have
\begin{align*}
M - f(S \cup R_{max})
&\leq (M - f(S)) \left(1 - \frac{|R_{max}| (1 - 11\eps)}{k} \right)\\
&\leq M \exp\left(- \frac{(1 - 11\eps)|S|}{k}\right) \exp\left(-\frac{(1 - 11\eps) |R_{max}|}{k} \right)\\
&= M \exp\left( - \frac{(1 - 11\eps)|S \cup R_{max}|}{k} \right)
\end{align*}
This completes the induction.

By the induction, in the first iteration where $f(S) \ge (1-\exp(-(1-12\eps)))M$ we must have $|S|\le (1-\eps)k +\ell < k$. Thus, the final solution satisfies the constraint $|S| \le k$ and has value $f(S) \ge (1-\exp(-(1-12\eps)))M$.
\end{proof}

\bibliographystyle{alpha}
\bibliography{submodular}

\newpage
\end{document}